\documentclass[a4paper,12pt,reqno]{amsart}
          \usepackage{amssymb}
	  \usepackage{amsmath}
          \usepackage{amsfonts}
          \usepackage[english]{babel}
          \usepackage[utf8]{inputenc}
\usepackage{bbold}
\usepackage{calrsfs}

\usepackage[margin=3cm]{geometry}

\usepackage{color}
 \usepackage[unicode,colorlinks,plainpages=false,hyperindex=true,bookmarksnumbered=true,bookmarksopen=true,pdfpagelabels]{hyperref}
\makeatletter
\renewcommand*{\p@section}{\S\,}
\renewcommand*{\p@subsection}{\S\,}
\renewcommand*{\p@subsubsection}{\S\,}
\makeatother

\usepackage{tikz}
 \usetikzlibrary{arrows}

\newtheorem{theorem}{Theorem}[section]
\newtheorem{lemma}[theorem]{Lemma}
\newtheorem{prop}[theorem]{Proposition}
\newtheorem{cor}[theorem]{Corollary}
\theoremstyle{definition}

\theoremstyle{remark}
\newtheorem{remark}[theorem]{Remark}

\numberwithin{equation}{section}

\newcommand{\la}{\label}

\newcommand{\Res}{\mathtt{Res}}

\newcommand{\End}{\mathtt{End}}

\newcommand{\id}{\mathrm{id}}
\newcommand{\into}{\hookrightarrow}
\newcommand{\ds}{\displaystyle}

\def\c{\mathbb{C}}
\def\R{\mathbb{R}}

\def\eh{\widehat{e}}

\def\Z{\mathbb{Z}}

\def\A{\mathcal{A}}
\def\H{\mathcal{H}}
\def\He{{H}}
\def\Hsf{\mathsf{H}}
\def\Asf{\mathsf A}
\def\Hsfh{\widehat{\mathsf{H}}}
\def\Asfh{\widehat{\mathsf A}}
\def\Hr{{H}^{(r)}}
\def\Lr{{L}^{(r)}}
\def\L{\mathsf L}
\def\Lh{\widehat{\mathsf L}}
\def\yr{{y}^{(r)}}
\def\pp{{\hat p}}

\def\Vc{V^{\vee}}
\def\Sc{S_n^\vee}

\def\D{\mathcal{D}}

\def\U{\mathcal U}

\def\Ah{A_\hbar}

\def\HH{\mathtt H}
\def\BB{\mathtt B}
\def\Irr{\mathtt{Irr}}

\begin{document}
%
%\begin{flushright}
%\textbf{Preliminary Version}
%\end{flushright}
%
%\vspace{1cm}
%
\title{Integrability of the Inozemtsev spin chain}

\author{Oleg Chalykh}
\address{School of Mathematics, University of Leeds, Leeds LS2 9JT, UK}
\email{o.chalykh@leeds.ac.uk}
%\date{\today}
%

%
\begin{abstract}
We show that the Inozemtsev spin chain is integrable. The conserved quantities (commuting Hamiltonians) are constructed using elliptic Dunkl operators. We also suggest a generalisation. 
\end{abstract}

\maketitle

\section{Introduction} The Inozemtsev quantum spin chain \cite{I1} is described by the Hamiltonian 
\begin{equation}\label{ham}
\H=\sum_{i<j}^n\wp\left(\frac{i-j}{n}\right)P_{ij}\,.
\end{equation}
Here $\wp(z)=\wp(z | 1, \tau)$ is the Weierstrass $\wp$-function with periods $1, \tau$ 
and $P_{ij}$ acts on the tensor power $U^{\otimes n}$ of  $U=\c^m$ by permuting the factors. This is viewed as a spin chain on $n$ sites, with each copy of $U$ representing local spin states, and with $\ds\wp\left(\frac{i-j}{n}\right)$ expressing the strength of interaction between the $i$th and $j$th sites. To show that it is integrable, one would like to find elements of $\End_\c(U^{\otimes n})$ commuting with $\H$. This is an old problem, see \cite{I2, DI08, DI, SS, FGL} for some of its history and related context. The trigonometric case -- with $\sin^{-2}(\pi z)$ in place of $\wp(z)$  
-- describes the celebrated Haldane--Shastry chain \cite{H, S} which is integrable and has the Yangian symmetry \cite{BGHP}. Nothing of that sort in known in the elliptic case so, as it stands, \eqref{ham}
falls outside the class of problems solved by Quantum Inverse Scattering 
or similar methods. 

Nevertheless, as will be demonstrated below, one can construct 
commuting Hamiltonians 
%$\H_2=\H$, $\H_3$, \dots, $\H_n$ 
from %Hamiltonians 
those of the quantum elliptic spin Calogero--Moser system,  by ``freezing''. This builds on (and elucidates) the original idea by Polychronakos \cite{P}, see also related works \cite{FM, MP, TH}.  
%We also construct additional Hamiltonians of the spin Calogero--Moser system, and use them to obtain further commuting Hamiltonians for the spin chain (as anticipated in \cite{DI}). 
Our main ingredients are the elliptic Dunkl operators and Cherednik algebra.
Along the way, we find a generalisation of the Inozemtsev spin chain.   

One would also like to calculate the eigenvalues and eigenvectors of $\H$, see \cite{KL} for the state of the art and references. Hopefully, the link to the elliptic Dunkl operators can shed some light onto this. Our methods apply, almost verbatim, to elliptic Calogero--Moser systems (and spin chains) for arbitrary root systems. They can also be extended to Calogero--Moser systems of $R$-matrix type and the corresponding spin chains from \cite{LOZ, SZ}, which will be done separately \cite{CL}. 

\medskip

\noindent\textbf{Acknowledgement.} I would like to thank Andrei Zotov for useful conversations. I am also grateful to Jules Lamers and Rob Klabbers for stimulating discussions and for waiting patiently for this paper to be written up.

\section{Rational Calogero--Moser system and Dunkl operators}
Here we recall the well-known relationship between the rational Calogero--Moser system, Dunkl operators and Cherednik algebra, see \cite{E} and references therein.

\subsection{} 
Let $V=\c^n$ be the $n$-dimensional complex Euclidean vector space with the standard basis $(e_i)$ and coordinates $(x_i)$.  The symmetric group $S_n$ acts on $V$ by permuting the basis vectors. This induces an $S_n$-action on the algebra $\c(V)$ of meromorphic functions in $n$ variables by $(w.f)(x)=f(w^{-1}.\,x)$ for $x\in V$, $w\in S_n$. The crossed product $\c(V)*S_n$ is formed by taking the vector space $\c(V)\otimes\c S_n$
with the multiplication $(f\otimes w)(f'\otimes w')=f(w.f')\otimes ww'$. 
Let $\D(V)$ denote the ring of differential operators on $V$ with meromorphic coefficients. It is generated by the partial derivatives $\partial_i=\frac{\partial}{\partial x_i}$ and operators of multiplication by $g\in \c(V)$. We have a natural action of $S_n$ on $\D(V)$, hence the crossed product $\D(V)*S_n$. As an algebra, $\D(V)*S_n$ is generated by its two subalgebras, $1\otimes\c S_n$ and $\D(V)\otimes 1$ which can be identified with $\c W$ and $\D(V)$, respectively. Using these identifications, we replace $a\otimes w$ by $aw$, so each element of $\D(V)*S_n$ is written uniquely as $a=\sum_{w\in S_n}a_ww$ with $a_w\in\D(V)$.

\subsection{} 
The \emph{quantum rational} Calogero--Moser (CM) system is described by the Hamiltonian
\begin{equation}\label{cm}
\Hr=\sum_{i=1}^n \hbar^2\partial_i^2-\sum_{i<j}^n\frac{2\kappa(\kappa-\hbar)}{(x_i-x_j)^2}\,.
\end{equation} 
Here $\kappa\in\c$ is the coupling constant and $\hbar\ne 0$ is a quantum parameter \footnote{To follow the standard conventions of quantum mechanics, $\hbar$ should be replaced $-\mathrm{i}\hbar$ and $\kappa$ should be purely imaginary.}. This is a completely integrable system, in the sense that it admits $n$ commuting Hamiltonians. 
These are constructed using the \emph{Dunkl operators}:    
\begin{equation}\label{yir}
\yr_i=\hbar \partial_i-\kappa\sum_{j\ne i}^n\frac{1}{x_i-x_j}s_{ij}\,,\qquad i=1,\dots, n\,.
\end{equation}
The key property of the these operators is their commutativity, $[\yr_i, \yr_j]=0$, and equivariance,
\begin{equation*}
s_{ij}\yr_i=\yr_js_{ij}\,,\qquad s_{ij}\yr_k=\yr_ks_{ij}\quad (k\ne i, j)\,.
\end{equation*} 
As a result, the assignment $\,e_i \mapsto \yr_i\,$
extends to a $S_n$-\,equivariant algebra map
\begin{equation}
\la{hom}
\c[V^*]\to \D(V)*S_n \,,\quad q \mapsto q(\yr) \ .
\end{equation}
Introduce the linear map 
\begin{equation}\label{res}
\Res:\ \D(V)*S_n\,\to\, \D(V)\,,\qquad \sum_{w\in S_n}a_ww\mapsto \sum_{w\in S_n}a_w\,.  
\end{equation}
Combining this map with \eqref{hom}, define
\begin{equation}\la{lp}
\Lr_q=\Res\,q(\yr)
\,,\qquad 
q\in \c[V^*]^{S_n}\,.
\end{equation}
The $S_n$-invariance of $q$ and the commutativity of the Dunkl operators imply that (1) each $\Lr_q$ is $S_n$-invariant, (2) $\Lr_q\Lr_{q'}=\Lr_{qq'}$ for any $q,q'\in \c[V^*]^{S_n}$, (3) the family $\{\Lr_q\,,\ q\in \c[V^*]^{S_n}\}$ is commutative.  

Taking $q=e_1^2+\dots +e_n^2\in S^2V\subset \c[V^*]$, a direct calculation shows that 
\begin{equation}\la{dcm}
q(\yr)=(\yr_1)^2+\dots +(\yr_n)^2=\sum_{i=1}^n \hbar^2\partial_i^2-\sum_{i<j}^n \frac{2}{(x_i-x_j)^2}\kappa(\kappa-\hbar s_\alpha)\,.
\end{equation} 
Hence, $\Lr_q=\Res\, q(\yr)$ in this case is precisely the CM Hamiltonian \eqref{cm}. Other symmetric combinations of $\yr_1, \dots, \yr_n$ produce higher commuting Hamiltonians.

\subsection{}\label{clcase} 
The classical limit corresponds to taking $\hbar\to 0$. More precisely, we 
view the Dunkl operators as elements of the algebra
\begin{equation*}
\Ah=\c[[\hbar]]\otimes \c(V)[\hat p_1, \dots, \hat p_n]\,,\qquad \hat p_i=\hbar\partial_{i}\,.
\end{equation*}
The quantum momenta $\hat p_i$ satisfy the relations $[\hat p_i, f]=\hbar\,\partial_{i}f$ for $f\in\c(V)$.
We have an algebra isomorphism
\begin{equation}\label{eta}
\eta_0:\ (\Ah*S_n)/\hbar(\Ah*S_n)\to A_0*S_n\,,\qquad f\mapsto f\,,\ \ \hat p_i\mapsto p_i\,,\ \ w\mapsto w\,,
\end{equation}
where $$A_0=\c(V)\otimes\c[V^*]=\c(V)[p_1, \dots, p_n]$$ is the classical version of $\Ah$.
Therefore, $\Ah$ and $\Ah*S_n$ are formal deformations of $A_0$ and $A_0*S_n$, respectively. 
Note that $A_0$ is commutative, with Poisson bracket defined by $\{\eta_0(a), \eta_0(b)\}=\eta_0(\hbar^{-1}[a,b])$ for $a,b\in \Ah$.
For any $a\in\Ah*S_n$, we call $\eta_0(a)$ the \emph{classical limit} of $a$. For instance, the classical limit of \eqref{yir} are the {\it classical Dunkl operators}
\begin{equation}\label{yirc}
\yr_{i,c}=p_i-\kappa\sum_{j\ne i}^n\frac{1}{x_i-x_j}s_{ij}\,,\qquad i=1,\dots, n\,.
\end{equation}
These are commuting, equivaraint elements of $A_0*S_n$, and we have a classical variant of \eqref{hom}:
\begin{equation}\label{homc}
\c[V^*] \to A_0*{S_n} \ ,\quad q \mapsto q(\yr_{c})\,.
\end{equation}  

\subsection{}
The classical limit of \eqref{cm} gives the Hamiltonian  
\begin{equation}\label{cmc}
\Hr_c=\sum_{i=1}^n p_i^2-\sum_{i<j}^n\frac{2\kappa^2}{(x_i-x_j)^2}\,.
\end{equation} 
Its complete integrability can be established similarly to the quantum case.  
Namely, using the classical variant of the map \eqref{res}, 
we set
\begin{equation}\la{lpc}
\Lr_{q,c}=\Res\,q(\yr_c)
\,,\qquad 
q\in \c[V^*]^{S_n}\,.
\end{equation}
By construction, $\Lr_{q,c}$ is the classical limit of $\Lr_q$. Hence, the family $\{\Lr_{q,c}\,,\ q\in \c[V^*]^{S_n}\}$ is Poisson commutative, i.e. it defines a classical integrable system. Passing to the classical limit in \eqref{dcm}, we get
\begin{equation}\la{dcmc}
(\yr_{1,c})^2+\dots +(\yr_{n,c})^2=\sum_{i=1}^n p_i^2-\sum_{i<j}^n \frac{2\kappa^2}{(x_i-x_j)^2}\,,
\end{equation}  
which is the Hamiltonian \eqref{cmc}. Note that in this case there is no need to apply $\Res$. This is true for the higher Hamiltonians as well:
\begin{lemma}[Lemma 2.2 \cite{EFMV}] \label{nores}
For any $q\in\c[V^*]^{S_n}$, we have $q(\yr_c)\in A_0$. Hence, $\Lr_{q,c}=q(\yr_c)$, i.e. the application of $\Res$ is not necessary. 
\end{lemma}
An analogous statement in the elliptic case will be important below. 

\subsection{} \label{cher}
By definition, the {\it rational Cherednik algebra}
$\HH_{\hbar, \kappa}$ of type $S_n$ is the subalgebra of $\D(V) *S_n$ generated by  $w\in S_n$, $x_i$ ($i=1, \dots, n$),
and the Dunkl operators \eqref{yir}. 
The {\it spherical subalgebra} of $ \HH_{\hbar, \kappa}
$ is defined as $e\, \HH_{\hbar, \kappa} \,e \,$, where 
\begin{equation}
e =\frac{1}{n!} \sum_{w \in S_n } w \,.
\end{equation}
Restricting \eqref{res} onto the spherical subalgebra, one obtains an algebra map
 \begin{equation}
\la{HC}
\Res:\,e \HH_{\hbar, \kappa} e \into \D(V)^{S_n}\,.
\end{equation}
whose image is denoted as $\BB_{\hbar ,\kappa}$. Elements of $B_{\hbar, \kappa}$ can be obtained by applying $\Res$ to symmetric combinations of $x_i$ and $\yr_i$. Inside $\BB_{0,\kappa}$ we have the commutative algebra of quantum rational CM Hamiltonians $\{\Lr_q\,,\, q\in\c[V^*]^{S_n}\}$. 

The classical limits are defined in a similar way. That is, $\HH_{0, \kappa}$ is the subalgebra of $A_0 *S_n$ generated by  $w\in S_n$, $x_i$ ($i=1, \dots, n$),
and the classical Dunkl operators \eqref{yirc}. The {spherical subalgebra} is defined as $e\, \HH_{0, \kappa} \,e \,$, and we have an algebra map 
 \begin{equation}
\la{HCc}
\Res:\,e \HH_{0, \kappa} e \into A_0^{S_n}\,.
\end{equation}
Its image, $\BB_{0, \kappa}$, is a Poisson subalgebra of $A_0^{S_n}$, containing a Poisson commutative algebra, $\{\Lr_{q,c}\,,\, q\in\c[V^*]^{S_n}\}$, of classical rational CM Hamiltonians. 

\section{Elliptic Dunkl operators and Calogero--Moser system}

\subsection{} 
The quantum elliptic CM system is described by the Hamiltonian
\begin{equation}\label{ecm}
\He=\hbar^2\sum_{i=1}^n\partial_i^2-2\kappa(\kappa-\hbar)\sum_{i<j}^n\wp(x_i-x_j)\,.
\end{equation} 
Here $\wp(z)=\wp(z | 1, \tau)$ is the Weierstrass $\wp$-function. The classical Hamiltonian is 
\begin{equation}\label{ecmc}
H_c=\sum_{i=1}^n p_i^2-2\kappa^2\sum_{i<j}^n\wp(x_i-x_j)\,. 
\end{equation} 
\subsection{} 
As in the rational case, both quantum and classical systems are completely integrable \cite{Ca, Pe, OP1, OP2, OS}.
Following \cite{EFMV}, this can be shown using \emph{elliptic Dunkl operators} \cite{BFV} which, in the quantum case, are    
\begin{equation}\label{yie}
y_i=\hbar\partial_i-\kappa\sum_{j\ne i}^n \phi(\lambda_i-\lambda_j, x_i-x_j)s_{ij}\,,\qquad \phi(\mu, z) =\frac{\sigma(z-\mu)}{\sigma(z)\sigma(-\mu)}\,.
\end{equation}
Here $\sigma(z)=\sigma(z |1, \tau)$ is the Weierstrass $\sigma$-function, and $\lambda_1, \dots, \lambda_n$ are auxiliary variables referred to as the \emph{spectral variables}. 
Hence, $y_i=y_i(\lambda)$ are $\lambda$-dependent elements of $\D(V)*S_n$.

For $\xi\in V$, we write $y_\xi=\sum_{i=1}^n \xi_iy_i$ .
Two main properties of the operators $y_\xi$ are their commutativity and equivariance: for all $\,\xi, \eta \in V\,$ and $ w \in S_n $,
\begin{equation}\la{eduprop}
 y_{\xi}\,y_{\eta} = y_{\eta}\,y_{\xi}\,,\qquad \,w\,y_\xi(\lambda) =
y_{w\xi}(w\lambda)\,w\,.
\end{equation}
Note that in the second relation the group action affects both $\xi$ and $\lambda$. As before, the assignment $\,\xi \mapsto y_\xi\,$
extends to an algebra map
\begin{equation}\label{ehom}
\c[V^*] \to \D(V)*S_n \ ,\quad q\mapsto q(y)\,.
\end{equation}
However, unlike in the rational case, this map is not $S_n$-equivariant and constructing commuting quantum Hamiltonians requires certain regularization. 

\subsection{}\label{lq} 
To explain the regularization procedure \cite{EFMV}, we extend the algebra map \eqref{ehom} by allowing polynomials with $\lambda$-dependent coefficients:
\begin{equation}\label{ehomex}
\c(V)\otimes \c[V^*] \to \D(V)*S_n \ ,\quad f\mapsto f(\lambda, y)\,.
\end{equation}
Next, recall the rational classical Hamiltonians $\Lr_{q,c}=q(\yr_c)$, $q\in\c[V^*]^{S_n}$. These are $S_n$-invariant elements of $A_0=\c(V)\otimes \c[V^*]$. 
\begin{theorem}[\cite{EFMV}]\label{efmv}
For $q\in\c[V^*]^{S_n}$, consider $\Lr_{q,c}\in A_0^{S_n}$. Identify $A_0$ with $\c(V)\otimes 
\c[V^*]$ in\eqref{ehomex} and obtain $\Lr_{q,c}(\lambda, y)\in \D(V)*W$ by applying \eqref{ehomex}, i.e. by replacing classical momenta in $\Lr_{q,c}$ with elliptic Dunkl operators. 

$(1)$ The elements $\Lr_{q,c}(\lambda, y)$ are regular near $\lambda=0$ and so have a well-defined limit at $\lambda=0$. 

$(2)$ Setting $L_q:=\Res \lim_{\lambda\to 0}\,\Lr_{q,c}(\lambda, y)$ defines an algebra embedding $\c[V^*]^{S_n}\to \D(V)^{S_n}$, $q\mapsto L_q$.
\end{theorem}

\subsection{} \label{ex}
To illustrate the theorem, let us follow a calculation in \cite{BFV}. Choose $\Lr_{q,c}$ to be the Hamiltonian $\Hr_c$ \eqref{cmc}. In this case,   
\begin{equation*}
\Hr_c(\lambda, y)=y_{1}^2+\dots +y_{n}^2-\sum_{i<j}\frac{2\kappa^2}{(\lambda_i-\lambda_j)^2}\,.
\end{equation*}
By a direct calculation, with the help of the identity $\phi(\mu, z)\phi(\mu, -z)=\wp(\mu)-\wp(z)$,
\begin{multline}\label{y2}
y_{1}^2+\dots +y_{n}^2=\hbar^2\sum_{i=1}^n\partial_i^2-2\hbar \kappa\sum_{i<j}^n \phi'(\lambda_i-\lambda_j, x_i-x_j)s_{ij}\\+
2\kappa^2\sum_{i<j}^n \left(\wp(\lambda_i-\lambda_j)-\wp(x_i-x_j) \right)\,.
\end{multline} 
Here $\phi'(\mu, z)=\frac{d}{dz}\phi(\mu, z)$. Using that 
\begin{equation*}
\wp(z)=\frac{1}{z^2}+o(z)\quad\text{as}\ z\to 0\,,\qquad \lim_{\mu\to 0}\phi'(\mu, z)=-\wp(z)\,,
\end{equation*}
one finds that
\begin{equation}\label{y2r}
\lim_{\lambda\to 0}\Hr_c(\lambda, y) = \hbar^2\sum_{i=1}^n\partial_i^2-\sum_{i<j}^n 2\kappa(\kappa-\hbar s_{ij})\wp(x_i-x_j)\,.
\end{equation}
Applying the map \eqref{res} gives the Hamiltonian \eqref{ecm}:
\begin{equation}\label{ecal}
\Res \lim_{\lambda\to 0} \Hr_c(\lambda, y)=H\,.
\end{equation}

\begin{remark}
For the Calogero--Moser particles, their total momentum $$P=p_1+\dots + p_n$$ is conserved. Replacing $p_i$ with the Dunkl operators $y_i$, one finds that
\begin{equation}
y_1+\dots +y_n= \hbar\partial_1+\dots +\hbar\partial_n\,,
\end{equation}  
reflecting conservation of the total momentum in the quantum system.
\end{remark}

\subsection{} 
By passing to the classical limit in the above constructions, one obtains Hamiltonians for the classical elliptic CM system. Namely, we think of the operators \eqref{yie} as elements of $\Ah*S_n$, so their classical limit are the following commuting elements of $A_0*S_n$:
\begin{equation}\label{yie}
y_{i,c}=p_i-\kappa\sum_{j\ne i}^n \phi(\lambda_i-\lambda_j, x_i-x_j)s_{ij}\,.
\end{equation} 
The classical limit of \eqref{ehomex} is the map
 \begin{equation}\label{ehomexc}
\c(V)\otimes \c[V^*] \to A_0*S_n \ ,\quad f\mapsto f(\lambda, y_c)\,.
\end{equation}
\begin{theorem}[\cite{EFMV}]\label{efmvc}
For $q\in\c[V^*]^{S_n}$, consider $\Lr_{q,c}\in A_0^{S_n}$. Identify $A_0$ with $\c(V)\otimes 
\c[V^*]$ in\eqref{ehomexc} and obtain $\Lr_{q,c}(\lambda, y_c)\in A_0*W$ by applying \eqref{ehomexc}, i.e. by replacing classical momenta in $\Lr_{q,c}$ with classical elliptic Dunkl operators. 

$(1)$ The elements $\Lr_{q,c}(\lambda, y_c)$ are regular near $\lambda=0$ and so have a well-defined limit at $\lambda=0$. 

$(2)$ Setting $L_{q, c}:=\Res \lim_{\lambda\to 0}\,\Lr_{q,c}(\lambda, y_c)$ defines a Poisson-commutative subalgebra $\{L_{q,c},\,\, q\in\c[V^*]^{S_n}\}$ in $A_0^{S_n}$.

$(3)$ In particular, similarly to \eqref{y2r},
\begin{equation*}
\Res \lim_{\lambda\to 0}\Hr_c(\lambda, y_c) = \sum_{i=1}^n p_i^2-2\kappa^2\sum_{i<j}^n \wp(x_i-x_j)\,, 
\end{equation*}
which is the Hamiltonian \eqref{ecmc}.
\end{theorem}
%We will refer to the families of elements $L_q$ and $L_{q,c}$ as higher CM Hamiltonians.
 
\subsection{}\label{hsf}
For later use, let us consider what happens when we substitute elliptic Dunkl operators into \emph{elliptic} Hamiltonians (rather than {rational} ones).  
\begin{prop}[Proposition 5.1 \cite{C19}]\label{prop19}
For $q\in\c[V^*]^{S_n}$, let $L_{q,c}\in A_0^{S_n}$ be the classcal Hamiltonian from Theorem \ref{efmvc}.  Then: 

$(1)$ $L_{q,c}(\lambda, y)$ is regular near $\lambda=0$, 

$(2)$ $L_{q,c}(\lambda, y_c)\in A_0*S_n$ does not depend on $\lambda$. Furthermore, expanding it into $\sum_{w\in S_n} a_ww$, we have $a_w=0$ for $w\ne\id$, i.e. $L_{q,c}(\lambda, y_c)\in A_0\subset A_0*S_n$. 
\end{prop}
Note that part $(2)$ of the proposition is an elliptic analogue of Lemma \ref{nores}. Now, viewing $L_{q,c}(\lambda, y)$ as elements of $\Ah*S_n$, denote 
\begin{equation}\label{hq}
\Hsf_q:=\Res \lim_{\lambda\to 0}L_{q,c}(\lambda, y)\,,\qquad q\in\c[V^*]^{S_n} .
\end{equation}
Clearly, $\Hsf_q$ belongs to $A_h^{S_n}$ and commutes with all higher CM Hamiltonians\footnote{In fact, one can show that $\Hsf_q=L_{q'}$ for some $q'\in\c[V^*]^{S_n}$, but we will not need this.}.
\begin{cor}\label{qcl}
For $q\in\c[V^*]^{S_n}$, we have 
\begin{equation}\label{exp}
\lim_{\lambda\to 0}L_{q,c}(\lambda, y)=\Hsf_q+\hbar\Asf_q\,,\qquad\text{for some}\ \Asf_q\in (\Ah*S_n)^{S_n}\,.
\end{equation} 
\end{cor}
Indeed, the classical limit of the l.h.s.~is $\lim_{\lambda\to 0}L_{q,c}(\lambda, y_c)=\Res\lim_{\lambda\to 0}L_{q,c}(\lambda, y_c)$ (the last equality is by part $(2)$ of Prop. \ref{prop19}). Hence, $\lim_{\lambda\to 0}L_{q,c}(\lambda, y)$ and $\Hsf_q$ have the same classical limit. \qed

\section{Spin Calogero--Moser system}\label{spincm}
The Dunkl operators can be used to construct a spin generalisation of the elliptic Calogero--Moser system. For the trigonometric case, see \cite{P1, BGHP, Che}.  
\subsection{} 
Consider the Hamiltonian  
\begin{equation}\label{spin}
\widehat H=\sum_{i=1}^n\hbar^2\partial_i^2-2\sum_{i<j}^n\wp(x_i-x_j)\kappa(\kappa-\hbar\widehat {s}_{ij})\,.
\end{equation} 
This is viewed as an element of the algebra $\D(V)\otimes \c S_n$, with the tensor-product sign omitted and with the ``hat'' symbol used to distinguish this from the crossed product $\D(V)*S_n$. Choosing an $S_n$-module $\U$ makes $\widehat H$ a matrix-valued differential operator, i.e. an element of $\D(V)\otimes \End_\c\U$, with $\widehat  s_{ij}$ acting on $\U$ (but not on $\c(V)$). For example, we may choose $\U=U^{\otimes n}$, $U=\c^m$, with $\widehat s_{ij}$ acting as the permutation $P_{ij}$ of the tensor factors. Our goal is to construct a commutative subalgebra of $\D(V)\otimes \c S_n$ containing $\widehat H$.

\subsection{} 
We begin with considerations as in \cite{P1, BGHP, Che}, rephrased for convenience. Let us enlarge the algebra $\D(V)*S_n$ 
by adding another copy of the group algebra:
\begin{equation*}
\A:=(\D(V)*S_n)\otimes \c S_n\,.
\end{equation*}  	
Elements of $\A$ can be uniquely written as
\begin{equation*}
a=\sum_{w,w'\in S_n}a_{ww'}w\otimes\widehat {w'}\quad\text{with}\ a_{ww'}\in\D(V)\,.
\end{equation*}
Define a linear map 
\begin{equation}\label{reshat}
\widehat{\Res}:\ \A\,\to\, \D(V)\otimes \c S_n\,,\qquad \sum_{w,w'\in S_n}a_{ww'}w\otimes\widehat {w'}
\mapsto \sum_{w,w'\in S_n}a_{ww'}\otimes\widehat {w'w^{-1}}\,. 
\end{equation}
Equivalently, $\widehat{\Res}(a)$ is the unique element $L_a\in\D(V)\otimes \c S_n$ such that
\begin{equation}
a\eh=L_a\eh\,,\qquad \eh:=\frac{1}{n!}\sum_{w\in S_n} w\otimes \widehat w\,.
\end{equation}
We have an $S_n$-action on $\A$ and $\D(V)\otimes \c S_n$ by conjugation, 
\begin{equation}\label{ac}
a\mapsto (w\otimes \widehat w)a (w^{-1}\otimes \widehat{w^{-1}})\qquad\forall\ w\in S_n\,.
\end{equation}
It is easy to check that the map $\widehat{\Res}$ is $S_n$-equivariant with respect to this action.   
\begin{lemma}\label{alg}
The restriction $\widehat{\Res}\,:\, \A^{S_n}\to (\D(V)\otimes \c S_n)^{S_n}$ of the map \eqref{reshat} is an algebra homomorphism.
\end{lemma}
Indeed, if $a\in \A^{S_n}$ and $L_a=\widehat\Res (a)$, then $a\eh=\eh a$ and $L_a\eh=\eh L_a$. Hence, for $a,b\in \A^{S_n}$ we have
$L_{ab}\eh=ab\eh =(a\eh)(b\eh)= (L_a\eh)(L_b\eh)=L_aL_b\eh$.
\qed

\subsection{} 
On $\D(V)*S_n\subset \A$, the action \eqref{ac} is simply $a\mapsto waw^{-1}$. Thus, restricting \eqref{reshat} further gives an algebra map 
\begin{equation}
\widehat\Res\,:\, (\D(V)*S_n)^{S_n}\to  (\D(V)\otimes \c S_n)^{S_n}\,,\qquad \sum_{w\in S_n}a_ww \mapsto  \sum_{w\in S_n}a_w\otimes \widehat{w^{-1}}\,,
\end{equation}
converting reflection-differential operators to $\c S_n$-valued differential operators. Recall that, by Theorem \ref{efmv}, we have well-defined elements $\lim_{\lambda\to 0} \Lr_{q,c}(\lambda, y)$, for every $q\in\c[V^*]^{S_n}$. Since $ w\Lr_{q,c}(\lambda, y)w^{-1}=\Lr_{q,c}(w\lambda, y)$, the limit at $\lambda=0$ is $S_n$-invariant. Also, since these elements are constructed from commuting Dunkl operators, they pairwise commute. Therefore, we have the following result.
\begin{prop}\label{prham}
For $q\in\c[V^*]^{S_n}$, define $\Lh_q=\widehat\Res \lim_{\lambda\to 0} \Lr_{q,c}(\lambda, y)$. The family $\{ \Lh_q\,,\, q\in\c[V^*]^{S_n}\}$ forms a commutative subalgebra in $(\D(V)\otimes \c S_n)^{S_n}$.
\end{prop}
Using our calculations in \ref{ex}, we find
\begin{align}
\Lh_q&=\hbar\partial_1+\dots+\hbar\partial_n\qquad&\text{for}\ q&=e_1+\dots + e_n\,,
\\\label{spcm}
 \Lh_q&=\sum_{i=1}^n \hbar^2\partial_i^2-2\sum_{i<j}^n\wp(x_i-x_j)\kappa(\kappa-\hbar\widehat {s}_{ij})\qquad&\text{for}\ q&=e_1^2+\dots+e_n^2\,.
\end{align}
Hence, the operator \eqref{spin} belongs to the constructed commutative family. 

\begin{remark}
The algebra $\D(V)*S_n$ has a nontrivial centre, spanned by the central idempotents $c_\pi\in \c S_n$, $\pi\in\Irr S_n$. The elements $\widehat{c_\pi}$ commute with any of the spin Hamiltonians. 
\end{remark}

\subsection{} \label{4.4}
Due to extra degrees of freedom in the spin system, one expects more commuting Hamiltonians. Recall the classical spherical subalgebra $\BB_{0, \kappa}$, a Poisson subalgebra of $A_0^{S_n}$, see \ref{cher}. Also, recall the map $A_0\to \D(V)*S_n$, $f\mapsto f(\lambda, y)$, see \eqref{ehomex}. We have the following result strengthening Proposition \ref{prham}. 

\begin{theorem}\label{sphham}
For any $f\in \BB_{0, \kappa}$, the element $f(\lambda, y)\in \D(V)*S_n$ is regular near $\lambda=0$. Define $\Lh_f:=\widehat \Res \lim_{\lambda\to 0} f(\lambda, y)$. Then the family $\{ \Lh_f\,,\, f\in \BB_{0, \kappa}\}$ forms a commutative subalgebra in $(\D(V)\otimes \c S_n)^{S_n}$.
\end{theorem}
The fact that $f(\lambda, y)$ are regular near $\lambda=0$ for any $f\in \BB_{0,\kappa}$ follows from the second proof of \cite[Theorem 3.1]{EFMV}. To be precise, in {\it loc. cit.} this is shown for $f=\Lr_{q,c}$, cf. Theorem \ref{efmv}. However, the proof applies verbatim to any $f\in \BB_{0, \kappa}$, see \cite[Sec. 5.3]{EFMV} for the details. The remaining statements are clear.  \qed

\subsection{}\label{4.5}
We can also substitute Dunkl operators into elliptic Hamiltonians (cf. Prop. \ref{prop19}) and obtain
\begin{equation}\label{hqspin}
\Hsfh_q:=\widehat\Res\lim_{\lambda\to 0} L_{q,c}(\lambda, y)\,,\qquad q\in\c[V^*]^{S_n}\,.
\end{equation}
These obviously commute between themselves and with any of the Hamiltonians 
$\Lh_f$, $f\in \BB_{0, \kappa}$.\footnote{One can show that, in fact, $\Hsf_q=\Lh_f$ for some $f\in \BB_{0, \kappa}$, but we will not need this result.} We will refer to the elements $\Lh_f$, $f\in \BB_{0, \kappa}$ as higher spin CM Hamiltonians, and to $\Hsfh_q$, $q\in\c[V^*]^{S_n}$ as \emph{principal} spin CM Hamiltonians.

\section{Inozemtsev quantum spin chain}\label{in}
We are now in a position to construct commuting Hamiltonians for the Inozemtsev spin chain \eqref{ham}. They will be obtained from the spin CM Hamiltonians by taking classical limits and evaluating at an equilibrium. In particular, this will provide a full justification for the ``freezing'' procedure \cite{P}. 
\subsection{} Consider the classical Hamiltonian $H$ \eqref{ecmc}, assuming $x_i, p_i\in \R$. 
Suppose $\tau\in \mathrm{i}\R_+$, so $\wp(z)$ is real for real $z$. The following fact is well known (it follows easily from the convexity of $\wp(z)$ for $0<z<1$).
\begin{lemma}  
In the region $x_1<\dots <x_n<x_1+1$, all possible equilibria $(x,p)$ of $H$ are of the form $x_i=c+\frac{i}{n}$ (with arbitrary $c\in\R$), $p_i=0$ ($i=1,\dots, n$). 
\end{lemma}
Write $(x^*, p^*)$ for an equlibrium of $H$ of the form $x_i=\frac{i}{n}$, $p_i=0$.

\begin{cor}\label{coreq} If a function $F=F(x,p)$ satisfies $\{F, H\}=\{F, x_1+\dots + x_n\}=0$ then $(x^*,p^*)$ 
is an equilibrium for $F$. 
\end{cor}
Indeed, the Hamiltonian flow defined by $F$ should map an equilibrium for $H$ to an equilibrium, at the same time preserving $x_1+\dots +x_n$. \qed 

\begin{remark}
Assuming that $F$ depends analytically on $\tau$, the above corollary remains valid for arbitrary, not necessarily real, $\tau$. This further implies that for any primitive period $\omega\in\Z+\Z\tau$, the point $(x, p)$ with $x_i=\frac{i}{n}\omega$, $p_i=0$ is an equilibrium for $F$.   
\end{remark}

\subsection{}
Recall CM Hamiltonians $\Hsf_q$, see \ref{hsf}. Let us restrict the choice of $q$ to 
$$\c[V^*]_0^{S_n}:=\c[V^*]^{S_n}/\langle e_1+\dots +e_n\rangle\,.$$
 In other words, we assume that $q$ is a symmetric polynomial of 
\begin{equation*}
\xi_i=e_i-\frac{1}{n}(e_1+\dots +e_n)\,,\qquad i=1,\dots, n\,. 
\end{equation*}
Let $y^0_i=y_{\xi_i}$ be the corresponding elliptic Dunkl operators. It is easy to check that $[y^0_i, x_1+\dots +x_n]=0$. Hence, the quantum CM Hamiltonians $\Hsf_q$ built from $y^0_i$ would commute with $x_1+\dots +x_n$, and their classical limit, $\Hsf_{q,c}$, would Poisson-commute with $x_1+\dots+x_n$, so we can apply the Corollary \ref{coreq}.
 \begin{prop}\label{jeq}
The classical CM Hamiltonians $\Hsf_{q,c}=\eta_0(\Hsf_q)$ with $q\in\c[V^*]^{S_n}_0$ have a joint equilibrium at $(x^*, p^*)$. Therefore, 
\begin{equation}
\{\Hsf_{q,c}, G\}=0\quad\text{at}\  (x,p)=(x^*, p^*)\qquad\forall\ G=G(x,p).
\end{equation}\qed
\end{prop}

\subsection{} Recall the commuting spin CM Hamiltonians $\Lh_f$, $\Hsfh_q$, see \ref{4.4}, \ref{4.5}. Let us view them as elements of $\Ah\otimes \c S_n$, so we can take classical limits using a map, similar to \eqref{eta}: 
$$\eta_0\,:\, \Ah\otimes \c S_n\to A_0\otimes \c S_n\,,\qquad f\mapsto f\,,\ \ \hat p_i\mapsto p_i\,,\ \ w\mapsto w\,.$$  
By applying $\widehat\Res$ to \eqref{exp}, we have
\begin{equation}\label{dec}
\Hsfh_q=\Hsf_q+\hbar\Asfh_q\qquad\text{with}\ \Asfh_q\in \Ah\otimes \c S_n\,.
\end{equation}
We define \emph{principal Inozemtsev Hamiltonians} by taking the classical limit of the spin part $\Asfh_q$ and evaluating it at equilibrium , i.e. 
\begin{equation}
\mathcal H_q:=\eta_0(\Asfh_q)|_{(x,p)=(x^*, p^*)}\,,\quad q\in\c[V^*]^{S_n}\,.
\end{equation}
For example, for $q=e_1^2+\dots+e_n^2$ one finds from \eqref{spcm}: 
\begin{align*}
\mathcal H_q&=2\kappa\sum_{i<j}^n\wp\left(\frac{i-j}{n}\right)(\widehat {s}_{ij}-\widehat\id)\,.
\end{align*}
The multiple of the identity can be removed so this is essentially the same as \eqref{ham}.
\begin{theorem} 
The elements $\mathcal H_q\in\c S_n$ pairwise commute.
\end{theorem}
\begin{proof}
First, note that for $q=e_1+\dots+e_n$, $\mathcal H_q=0$. Hence, it is sufficient to prove commutativity for
$q\in\c[V^*]^{S_n}_0$. We need the following elementary lemma.
\begin{lemma}\label{lemma1}
Let $B=B_0+\hbar B_1$, $C=C_0+\hbar C_1$ be elements of $\Ah\otimes \c S_n$ with $B_0, C_0\in \Ah\subset \Ah\otimes \c S_n$. Denote by $b_0,  b_1, c_0, c_1$ the classical limits of $B_0, B_1, C_0, C_1$. Then, assuming $[B_0, C_0]=0$, one has $[B,C]\in\hbar^2\Ah\otimes \c S_n$, and
\begin{equation*}
\eta_0(\hbar^{-2}[B,C])=\{b_0, c_1\}-\{c_0, b_1\}+[b_1, c_1]\,.
\end{equation*}
Here the Poisson bracket in the r.h.s.~is understood as $\{a, b\otimes\widehat w\}:=\{a,b\}\otimes\widehat w$ for $a,b\in A_0$. %\qed
\end{lemma}
Now, given $q,q'\in\c[V^*]^{S_n}_0$, we apply the lemma to $B=\Hsfh_q$, $C=\Hsfh_{q'}$ and use the decomposition \eqref{dec}. Since $[B,C]=0$ in this case, we get $0=\{b_0, c_1\}-\{c_0, b_1\}+[b_1, c_1]$. The first two terms in the r.h.s.~vanish at $(x,p)=(x^*, p^*)$, by Prop. \ref{jeq}. As a result, $[b_1, c_1]$ vanishes at $(x^*, p^*)$, i.e. $[\mathcal H_q, \mathcal H_{q'}]=0$.\footnote{When finalising this paper, we learned about the work \cite{LRS} where a similar idea was used for the Haldane--Shastry chain.}
\end{proof}

\subsection{}
We proceed to define \emph{higher Inozemtsev Hamiltonians} by
\begin{equation}
\mathcal H_f:=\eta_0(\Lh_f)|_{(x,p)=(x^*, p^*)}\,,\quad f\in \BB_{0, \kappa}\,.
\end{equation}
\begin{theorem} 
The elements $\mathcal H_f$ pairwise commute, as well as commuting with the principal Hamiltonians $\mathcal H_q$.
\end{theorem}
The fact that $[\mathcal H_f, \mathcal H_{f'}]=0$ for $f, f'\in \BB_{0,\kappa}$ is immediate from $[\Lh_f, \Lh_{f'}]=0$ (evaluation at $(x^*, p^*)$ plays no role here).
The fact that $[\mathcal H_q, \mathcal H_{f}]=0$ follows from the following lemma (obtained by setting $C_0=0$ in Lemma \ref{lemma1}). 
\begin{lemma}
Let $B=B_0+\hbar B_1$, $C$ be elements of $\Ah\otimes\c S_n$ with $B_0\in \Ah\subset \Ah\otimes \c S_n$. Denote by $b_0,  b_1, c$ the classical limits of $B_0, B_1, C$. Then one has $[B,C]\in\hbar\Ah\otimes \c S_n$, and
\begin{equation*}
\eta_0(\hbar^{-1}[B,C])=\{b_0, c\}+[b_1, c]\,.
\end{equation*}
\end{lemma}
\qed

\begin{remark}
The above results and constructions remain valid for elliptic Calogero--Moser systems associated for any root system (including the $BC_n$ version) - all the necessary ingredients can already be found in the cited sources. In general, equilibria for these systems will not have such a simple form as above, which makes the corresponding spin chains less natural. In some cases, however, simple equilibria configurations are possible, see, e.g., \cite{IS, BPS}.    
\end{remark}

\section{Explicit examples}

Here we calculate two of the Hamiltonians explicitly, as an illustration and to compare with known results. 

\subsection{} We use the following shorthand notation: $\phi_{ij}^{ab}=\phi(\lambda_a-\lambda_b, x_i-x_j)$, $\phi_{ij}:=\phi_{ij}^{ij}$, $h_{ij}^{ab}=-\phi'(\lambda_a-\lambda_b, x_i-x_j)$, $h_{ij}:=h_{ij}^{ij}$, $u_{ij}=\wp(x_i-x_j)$, $u^{ab}=\wp(\lambda_a-\lambda_b)$, $\theta_i=\sum_{j\ne i} \phi_{ij}s_{ij}$. With these notations, the Dunkl operators are
$
y_i=\pp_i-\kappa\theta_i
$. Let us substitute them into the cubic classical Hamiltonian $L=L_q$ with $q=\sum_{i<j<k}e_ie_je_k$:
\begin{equation*}\label{L}
L=\sum_{i<j<k}\left(p_ip_jp_k+\kappa^2u^{ij}p_k+\kappa^2u^{ik}p_j+\kappa^2u^{jk}p_i\right)\,.
\end{equation*}
Thus, we need to evaluate 
\begin{equation}\label{sum1}
L(\lambda, y)=\frac{1}{6}\sum_{i\ne j\ne k}y_iy_jy_k+\kappa^2\sum_{i\ne j\ne k} u^{ij}y_k\,.
\end{equation}
First, consider 
$$
y_iy_jy_k=\pp_i\pp_j\pp_k-\kappa(\pp_i\pp_j\theta_k+\pp_i\theta_j\pp_k+\theta_i\pp_j\pp_k)+\kappa^2(\theta_i\theta_j\pp_k+\theta_i\pp_j\theta_k+\pp_i\theta_j\theta_k)-\kappa^3\theta_i\theta_j\theta_k\,.
$$
When expanding each $\theta_a$ in terms of $\phi_{ab}s_{ab}$ and taking the sum of these over all $i\ne j\ne k$, some terms cancel due to $\phi_{ia}=-\phi_{ai}$. For example, the term $\pp_i\pp_j\phi_{ka}s_{ka}$ (with $a\ne i,j,k$) coming from $y_iy_jy_k$ cancels with a similar term  $\pp_i\pp_j\phi_{ak}s_{ak}$ coming from $y_iy_jy_a$. In addition, some terms linear in momenta cancel due to Fay's identity,
$$
\phi_{ij}^{ik}\phi_{jk}=\phi_{ik}^{jk}\phi_{ij}+\phi_{jk}^{ji}\phi_{ik}\qquad (i\ne j\ne k)\,,
$$  
which implies that for $a\ne i,j$
$$
\phi_{ia}s_{ia}\phi_{ja}s_{ja} +\ \text{(cyclic permutations of $i, j, a$)}\ =0\,.
$$
Similarly, when expanding $\theta_i\theta_j\theta_k$, we may neglect terms with indices outside $\{i,j,k\}$ because of the previous identity and its higher analogue, valid for $a\ne i,j,k$:
$$
\phi_{ia}s_{ia}\phi_{ja}s_{ja}\phi_{ka}s_{ka}+\ \text{(cyclic permutations of $i, j, k, a$)}\ =0\,.
$$    
As a result, when calculating $\sum_{i\ne j\ne k}y_iy_jy_k$, it is sufficient to keep in each $y_iy_jy_k$ only the terms with indices from $\{i,j,k\}$ set (i.e.~calculating $y_iy_jy_k$ \emph{as if} it was just for three particles at $x_i, x_j, x_k$.)
By the same token, in the second sum in \eqref{sum1}, we have a cancellation of, e.g., $u^{ij}\phi_{ka}s_{ka}$ with $u^{ij}\phi_{ak}s_{ak}$ (if $a\ne i,j,k$). As a result, we may replace 
$$
u^{ij}y_k \mapsto u^{ij}(\pp_k-\kappa\phi_{ki}s_{ki}-\kappa\phi_{kj}s_{kj})\,.
$$ 
Altogether, \eqref{sum1} can be replaced with $\sum_{i<j<k}H_{ijk}$ where
\begin{multline*}\label{sum2}
H_{ijk}=(\pp_i-\kappa\phi_{ij}s_{ij}-\kappa\phi_{ik}s_{ik})(\pp_j-\kappa\phi_{ji}s_{ji}-\kappa\phi_{jk}s_{jk})(\pp_k-\kappa\phi_{ki}s_{ki}-\kappa\phi_{kj}s_{kj})
\\+\kappa^2\left\{u^{ij}(\pp_k-\kappa\phi_{ki}s_{ki}-\kappa\phi_{kj}s_{kj})+\ \text{(cyclic permutations of $i, j, k$)}\right\}.
\end{multline*}
With the help of Fay's identity and $\phi_{ij}^{ij}\phi_{ij}^{ji}=u_{ij}-u^{ij}$, this becomes 
\begin{align*}
H_{ijk}=& \pp_i\pp_j\pp_k+\kappa^2u_{ij}\pp_k+\kappa^2u_{ik}\pp_j+\kappa^2u_{jk}\pp_i
\\&-\hbar \kappa (h_{ij}\pp_ks_{ij}+h_{jk}\pp_is_{jk}+h_{ki}\pp_js_{ki})
\\&+\hbar\kappa^2\left\{(h_{ij}\phi_{ki}^{kj}+ h_{ij}^{ik}\phi_{jk})+\ \text{(cyclic permutations of $i, j, k$)}\right\}(ijk)
\\&+\hbar\kappa^2\left\{(h_{ji}\phi_{kj}^{ki}+h_{ji}^{jk}\phi_{ik})+\ \text{(cyclic permutations of $i, j, k$)}\right\}(kji)\,.
\end{align*} 
Taking $\lambda\to 0$ gives 
\begin{align*}
H_{ijk}=& \pp_i\pp_j\pp_k+\kappa^2u_{ij}\pp_k+\kappa^2u_{ik}\pp_j+\kappa^2u_{jk}\pp_i
\\&-\hbar \kappa (u_{ij}\pp_ks_{ij}+ u_{jk}\pp_is_{jk}+ u_{ki}\pp_js_{ki})
\\&+\hbar\kappa^2\left\{(u_{ij}+u_{jk}+u_{ki})(\zeta_{ij}+\zeta_{jk}+\zeta_{ki})+ \frac12 (u'_{ij}+u'_{jk}+u'_{ki})\right\}{(ijk)}
\\&+\hbar\kappa^2\left\{(u_{kj}+u_{ji}+u_{ik})(\zeta_{kj}+\zeta_{ji}+\zeta_{ik})+ \frac12 (u'_{kj}+u'_{ji}+u'_{ik})\right\}{(kji)}\,.
\end{align*}
Applying $\widehat\Res$ gives a spin Hamiltonian that matches the expression for $I_2$ in \cite{DI}. Furthermore, picking terms of order $\hbar$ and evaluating at an equilibrium in the classical limit, we get a principal Inozemtsev Hamiltonian
\begin{align*}
\mathcal H_3=&\sum_{i<j<k} \left\{(\wp_{ij}+\wp_{jk}+\wp_{ki})(\zeta_{ij}+\zeta_{jk}+\zeta_{ki})+ \frac12 (\wp'_{ij}+\wp'_{jk}+\wp'_{ki})\right\}(\widehat{(ijk)}-\widehat{(kji)})\,.
\end{align*}
Here $\wp_{ij}=\wp(\frac{i-j}{n})$, $\zeta_{ij}=\zeta(\frac{i-j}{n})$, $\wp'_{ij}=\wp'(\frac{i-j}{n})$. Note that by the addition formula for $\wp(z)$, $\wp_{ij}+\wp_{jk}+\wp_{ki}=(\zeta_{ij}+\zeta_{jk}+\zeta_{ki})^2$, so $\mathcal H_3$ matches $\hat{I}_2$ from \cite{I1}.

\subsection{} Let us also consider higher Hamiltonians obtained from $f\in \BB_{0,\kappa}$. 
First, let us try 
$$
f=\Res \sum_{i=1}^n x_i\yr_{i,c}=\sum_{i=1}^n x_ip_i\,.
$$ 
Replacing the classical momenta with elliptic Dunkl operators, we get
\begin{equation*}
f(\lambda, y)=\sum_{i=1}^n \lambda_iy_i=\sum_{i=1}^n \lambda_i(\pp_i+\kappa\sum_{j\ne i}\phi_{ij}s_{ij})\,.
\end{equation*}
We have 
\begin{equation}\label{as}
\phi_{ij}=-\frac{1}{\lambda_{i}-\lambda_j}+\zeta(x_i-x_j)+o(1)
\quad\text{near}\ \lambda=0\,.
\end{equation}
Using this in the previous formula, we find that
$$
\lim_{\lambda\to 0} f(\lambda, y)=-\sum_{i<j}s_{ij}\,,
$$
which commutes with other spin Hamiltonians for trivial reasons. A more interesting choice is 
$$
f=\sum_{i\ne j\ne k}x_i\left(p_jp_k+\frac{\kappa^2}{x_{jk}^2}\right)\,.
$$ 
Such $f$ can be obtained by applying $\Res$ to $\sum_{i\ne j\ne k} x_i\yr_{j,c}\yr_{k,c}$ (and subtracting a multiple of $p_1+\dots +p_n$).
We then need to consider 
$$
f(\lambda, y)=\sum_{i\ne j\ne k}\lambda_i\left(y_jy_k+\frac{\kappa^2}{\lambda_{jk}^2}\right)\,.
$$ 
Since we are interested in the $\lambda\to 0$ limit, we may expand $y_jy_k$ using \eqref{as} and keep track only of the terms singular in $\lambda$.  After a straightforward rearrangement, the limiting form of $f(\lambda, y)$ is found as
$$
-\kappa\sum_{i\ne j\ne k} \pp_is_{jk}-\kappa^2\sum_{i<j<k}\left\{\zeta(x_i-x_j)+\zeta(x_j-x_k)+\zeta(x_k-x_i)\right\}((ijk)-(kji))\,.
$$
Applying $\widehat\Res$ gives an additional spin CM Hamiltonian 
\begin{equation*}
\Hsfh_f=-\kappa\sum_{i\ne j\ne k} \pp_i\widehat{s_{jk}}+\kappa^2\sum_{i<j<k}\left\{\zeta(x_i-x_j)+\zeta(x_j-x_k)+\zeta(x_k-x_i)\right\}(\widehat{(ijk)}-\widehat{(kji)})\,.
\end{equation*}
This matches the quantum spin Hamiltonian $I_1$ found in \cite{DI}. 
Taking classical limit and evaluating at equilibrium, we obtain (setting $\kappa=1$)
$$
\mathcal H_f=\sum_{i<j<k}\left\{\zeta_{ij}+\zeta_{jk}+\zeta_{ki}\right\}(\widehat{(ijk)}-\widehat{(kji)})\,,\qquad\zeta_{ij}=\zeta\left(\frac{i-j}{n}\right)\,.
$$ 
This coincides with $\hat{I}_1$ from \cite{I1}. Note that it is surprisingly difficult to check that $[\hat{I}_1, \hat{I}_2]=0$ by a direct calculation, see \cite{DI08}. 

\begin{remark}
For calculating $\mathcal H_f$ with $f(x,p)\in B_{0, \kappa}$, we may take classical limit and evaluate at equilibrium immediately, i.e. substitute $x_i\mapsto \lambda$, $p_i\mapsto \theta_i=\sum_{j\ne i}\phi_{ij}s_{ij}$. However, finding the limiting value at $\lambda\to 0$ seems a lengthy process in general. 
%In the above example this gives
%$$ 
%\sum_{i\ne j\ne k}\lambda_i\left(\kappa^2\theta_j\theta_k+\frac{\kappa^2}{\lambda_{jk}^2}\right)\,,\qquad \theta_a=\sum_{b\ne a}\phi_{ab}s_{ab}\,.
%$$
%We may set $\kappa=1$ and it only remains to calculate the limit as $\lambda\to 0$. 
\end{remark}

\section{A generalisation}

\subsection{}
Consider the following modification of the Hamiltonian \eqref{ham}:
\begin{equation}\label{hamex}
\H^\vee=\sum_{i<j}^n\phi'(\lambda_i-\lambda_j, \frac{i-j}{n})P_{ij}\,.
\end{equation}
Here $\phi'(\mu, z)=\frac{d}{dz}\phi(\mu,z)$ is considered as a formal series in $\mu$,
\begin{equation*}
\phi'(\mu,z)=\sum_{i=0}^\infty c_i\mu^i\,,\qquad c_0=-\wp(z)\,,\ c_1=\wp'(z)+2\zeta(z)\wp(z)\,,\dots\,,
\end{equation*} 
making $\phi'(\lambda_i-\lambda_j, \frac{i-j}{n})$ a formal series in $\lambda_i, \lambda_j$.
The operator \eqref{hamex} is viewed as acting on the tensor product $U^{\otimes n}$ of $U=\c^m\otimes \c[[s]]$, with $P_{ij}$ permuting the factors, and with the $i$th spectral variable $\lambda_i$ acting as $s^{(i)}$, i.e. as multiplication by $s$ on the $i$th factor and identity on the others. Note that the action of $\H$ naturally descends onto $(U/s^{l+1}U)^{\otimes n}$, for any $l\ge 0$. In particular, taking $l=0$ recovers the original Hamiltonian \eqref{ham}. 

\subsection{} To demonstrate the integrability of the Hamiltonian \eqref{hamex}, we first construct a dynamical spin model, modifying the approach in \ref{spincm} to incorporate dependence on the spectral variables $\lambda$. It will be convenient to introduce an auxiliary space $\Vc\cong V$ with $\lambda\in \Vc$, and denote by $\Sc$ a copy of the symmetric group acting on $\Vc$.  
Consider the algebra $\D^\vee(V)$ of differential operators on $V$ whose coefficients also depend on $\lambda$; it is generated by (operators of multiplication by) functions $f\in \c(V\times \Vc)$ and the derivations $\partial_i=\partial/\partial x_i$. The group $S_n\times \Sc$ acts on $V\times \Vc$ and on $\D^\vee(V)$, so we can form a crossed product 
\begin{equation*}
\A^\vee:=\D^\vee(V)*(S_n\times\Sc)\,.
\end{equation*}  	
Elements of $\A^\vee$ can be uniquely written as
\begin{equation*}
a=\sum_{w_1, w_2\in S_n}a_{w_1w_2}w_1\otimes w_2^{\vee}\quad\text{with}\ a_{w_1w_2}\in\D^\vee(V)\,.
\end{equation*}
Define a linear map 
\begin{align}\nonumber
{\Res}^\vee&:\ \A^\vee\,\to\, \D^\vee(V)*\Sc\,,\\ 
\label{reshatl}
&\sum_{w_1,w_2\in S_n}a_{w_1w_2}w_1\otimes w_2^\vee\ 
\mapsto \sum_{w_1,w_2\in S_n}a_{w_1w_2}(w_2w_1^{-1})^\vee\,. 
\end{align}
Equivalently, ${\Res}^\vee(a)$ is the unique element $L_a\in\D^\vee(V)*\Sc$ such that
\begin{equation}
ae^\vee=L_ae^\vee\,,\qquad e^\vee:=\frac{1}{n!}\sum_{w\in S_n} w\otimes w^\vee\,.
\end{equation}
We have an $S_n$-action on $\A^\vee$ and $\D^\vee(V)*\Sc$ by conjugation, 
\begin{equation}\label{acc}
a\mapsto (w\otimes w^\vee)a (w\otimes w^\vee)^{-1}\qquad\forall\ w\in S_n\,.
\end{equation}
It is easy to check that the map ${\Res}^\vee$ is $S_n$-equivariant with respect to this action.   
\begin{lemma}
The restriction ${\Res}^\vee\,:\, (\A^\vee)^{S_n}\to (\D^\vee(V)*\Sc)^{S_n}$ of the map \eqref{reshat} is an algebra homomorphism.
\end{lemma}
Proof is identical to that of Lemma \ref{alg}. 
\qed

\subsection{} 
Restricting \eqref{reshatl} further gives an algebra map 
\begin{align*}
\Res^\vee\,:&\, (\D^\vee(V)*S_n)^{S_n}\to  (\D^\vee(V)*\Sc)^{S_n}\,,\\ 
&\sum_{w\in S_n}a_ww \mapsto  \sum_{w\in S_n}a_w(w^{-1})^\vee\,.
\end{align*}
Comparing with a similar map in \ref{spincm}, it takes elements invariant under ``diagonal''  $S_n$-action (on both $x, \lambda$). Another difference is the crossed-product structure in the resulting operator so it acts on the spectral variables.

As for examples of elements in $(\D^\vee(V)*S_n)^{S_n}$, we can substitute elliptic Dunkl operators into any $S_n$-invariant function on $\Vc\otimes V^*$. However, we would like the result to be regular near $\lambda=0$ so we can expand the coefficients into a formal series on $\lambda$ as we did for \eqref{hamex}. This motivates the following  definitions:
\begin{align}\label{defh}
\L^\vee_f:=\Res^\vee f(\lambda, y)\,,\quad
\Hsf^\vee_q:=\Res^\vee L_{q,c}(\lambda, y)\,,\qquad \text{with}\  f\in \BB_{0,\kappa}\,, q\in\c[V^*]^{S_n}\,.
\end{align}   
These are similar to the definitions in \ref{spincm}, but one important difference is that we do not take limit $\lambda\to 0$.

\begin{theorem}\label{sphhamc}
The elements $\L^\vee_f$ and $\Hsf^\vee_q$ are pairwise commuting, $S_n$-equivariant  elements of $\D^\vee(V)*\Sc$, regular near $\lambda=0$.
\end{theorem}
Proofs are the same as in \ref{spincm}. \qed

As an example, using the calculations in \ref{ex}, we find that for $q=e_1^2+\dots+e_n^2$,
\begin{equation}\label{hq2l}
\Hsf^\vee_q=\sum_{i=1}^n\hbar^2\partial_i^2-2\kappa^2\sum_{i<j}^n\wp(x_i-x_j)
-2\hbar \kappa\sum_{i<j}^n \phi'(\lambda_i-\lambda_j, x_i-x_j)s^\vee_{ij}\,.
\end{equation}
This is a generalisation of the spin Calogero--Moser Hamiltonian \eqref{spin}. 
Choosing any $\c[[\Vc]]*\Sc$-module $\mathcal U$, and expanding $\Hsf^\vee_q$ into a formal series in $\lambda$ makes it a differential operator acting on $\mathcal U$-valued functions $\varphi: V\to \mathcal U$.
For example, we can choose $\mathcal U=U^{\otimes n}$ with $U=\c^m\otimes\c[[s]]$ and with $\lambda_i$ acting by $s^{(i)}$, in the same way as it was done for the Hamiltonian \eqref{hamex}. 

\subsection{} Now the commuting Hamiltonians for the spin chain \eqref{hamex} can be constructed in the same way as before. We omit the proofs as they are identical to those in \ref{in} and use evaluation at an equilibrium $(x^*, p^*)$. 

Let us view the Hamiltonians \eqref{defh} as elements of $\Ah * \Sc$, so we can take classical limits. We have a counterpart of \eqref{dec}:
\begin{equation}\label{decl}
\Hsf^\vee_q=\Hsf_q+\hbar\Asf^\vee_q\qquad\text{with}\ \Asf^\vee_q\in \Ah*\Sc\,.
\end{equation}
We define \emph{principal Hamiltonians} by
\begin{equation}
\mathcal H_q^\vee:=\eta_0(\Asf^\vee_q)|_{(x,p)=(x^*, p^*)}\,,\quad q\in\c[V^*]^{S_n}\,.
\end{equation}
For example, for $q=e_1^2+\dots+e_n^2$ one finds from \eqref{hq2l}:
\begin{equation}\label{hamex1}
\mathcal H_q^\vee=-2\kappa\sum_{i<j}^n\phi'\left(\lambda_i-\lambda_j, \frac{i-j}{n}\right){s}_{ij}^\vee-2\kappa\sum_{i<j}^n \wp\left(\frac{i-j}{n}\right)
\id^\vee\,.
\end{equation}
The multiple of the identity can be removed, giving 
 \begin{equation}\label{hamex2}
\H^\vee=\sum_{i<j}^n\phi'\left(\lambda_i-\lambda_j, \frac{i-j}{n}\right){s}_{ij}^\vee\,.
\end{equation}
We also define \emph{higher Hamiltonians} by
\begin{equation}
\mathcal H^\vee_f:=\eta_0(\L^\vee_f)|_{(x,p)=(x^*, p^*)}\,,\quad f\in \BB_{0,\kappa}\,.
\end{equation}
\begin{theorem} \label{intc}
The elements $\H_q^\vee$, $\H_f^\vee$ with $q\in\c[V^*]^{S_n}$, $f\in \BB_{0,\kappa}$ form a commutative family in $\c[[\Vc]]*S_n^\vee$.
\end{theorem}
Choosing $\mathcal U=U^{\otimes n}$ with $U=\c^m\otimes\c[[s]]$ as a representation of $\c[[\Vc]]*S_n^\vee$ identifies \eqref{hamex2} with the Hamiltonian \eqref{hamex}, hence, it is integrable.

%These elements would act in any chosen representation $\mathcal U$ of $\c[[\Vc]]*\Sc$. Choosing $\mathcal U=U^{\otimes n}$ with $U=\c^m\otimes\c[[s]]$, as before, identifies \eqref{hamex2} with the Hamiltonian \eqref{hamex}; hence, Theorem \ref{intc} demonstrates its integrability. \qed

%\begin{remark}
%The elements of the form ${c_\pi}^\vee$ for any of the central idempotents $c_\pi\in \c S_n$ are central in $\D_\lambda(V)*\Sc$, so they will commute with any of the spectral Inozemtsev Hamiltonians. In particular, we can remove the second sum in \eqref{hamex1}, and rescale to get
%\end{remark}

\end{document}